\documentclass[letterpaper, 10 pt, conference]{ieeeconf}  

\IEEEoverridecommandlockouts                              
\usepackage{comment}
\usepackage{amsmath, amssymb, amsfonts,bm}
\usepackage{mathrsfs}
\usepackage[mathcal]{euscript}

\usepackage{enumerate}
\usepackage{mathbbol}

\usepackage{float}

\usepackage{xcolor,graphicx,epstopdf}
\usepackage{dsfont}
\usepackage{hyperref}
\usepackage{caption}
\usepackage{subcaption}
\usepackage{cite} 

\usepackage{tikz-cd}
\usepackage{mathtools}

\usepackage[subfigure]{tocloft}
\usepackage{graphicx}
\usepackage{float}
\usepackage{epsfig}

\usepackage{algorithm}

\usepackage{tikz}
\usepackage{pgfplots}
\usetikzlibrary{positioning,calc}

\usepackage{xcolor}
\definecolor{midnightblue}{rgb}{0, 0.14,0.55}
\usepackage[normalem]{ulem}

\newcommand*\mcapinn[2]{\vcenter{\hbox{$\mathsurround=0pt
			\ifx\displaystyle#1\textstyle\else#1\fi\bigcap$}}}

\newcommand*\mcupinn[2]{\vcenter{\hbox{$\mathsurround=0pt
			\ifx\displaystyle#1\textstyle\else#1\fi\bigcup$}}}

\newtheorem{theorem}{Theorem}
\newtheorem{definition}{Definition}

\newtheorem{remark}{Remark}

\title{\LARGE \bf
	An Electricity Market with Reactive Power Trading: Incorporating Dynamic Operating Envelopes
}

\author{Zeinab Salehi, Elizabeth L. Ratnam, Yijun Chen, Ian R. Petersen, Guodong Shi, and Duncan S. Callaway
	\thanks{This work was supported by the Australian Research Council under grants DP190102158, DP190103615, LP210200473, and DP230101014.}  
	\thanks{Zeinab Salehi and Ian R. Petersen are with the School of Engineering, Australian National University, Canberra, ACT 2601, Australia (e-mail:  zeinab.salehi@anu.edu.au; ian.petersen@anu.edu.au).}
\thanks{Elizabeth L. Ratnam is with the Department of Electrical and Computer Systems Engineering, Monash University, Melbourne, VIC 3168, Australia (e-mail:  liz.ratnam@monash.edu).}
\thanks{Yijun Chen is with the Department of Electrical and Electronic Engineering,  University of Melbourne,  VIC 3052, Australia (e-mail:  yijun.chen.1@unimelb.edu.au).}
\thanks{Guodong Shi is with the  Australian Centre for Robotics, School of Aerospace, Mechanical and Mechatronic Engineering, The University of Sydney, NSW 2050, Australia (e-mail: guodong.shi@sydney.edu.au).}
\thanks{Duncan S. Callaway is with the Energy and Resources Group, University of California, Berkeley, CA 94720 USA (e-mail: dcal@berkeley.edu).}
}

\begin{document}

	\maketitle
	\thispagestyle{empty}
	\pagestyle{empty}

	\begin{abstract}
		Electricity market design that accounts for grid constraints such as voltage and thermal limits at the distribution level can increase opportunities for the grid integration of Distributed Energy Resources (DERs). In this paper, we consider rooftop solar backed by battery storage connected to a distribution grid. We design an electricity market to support customers sharing rooftop generation in excess of their energy demand, where customers earn a profit through peer-to-peer (P2P) energy trading. Our proposed electricity market also incorporates P2P reactive power trading to improve the voltage  profile across a distribution feeder. We formulate the electricity market as an optimization-based problem, where voltage and thermal limits across a feeder are managed through the assignment of customer-specific dynamic operating envelopes (DOEs). The electricity market equilibrium is referred to as a competitive equilibrium, which is equivalent to a Nash equilibrium in a standard game. Our proposed market design is benchmarked using the IEEE 13-node test feeder.
	\end{abstract}

\section{Introduction}
Power systems are transitioning from centralized fossil fuel-based generation to renewable-based operations --- with a more geographically distributed energy mix. An element of the energy transition is the recent rapid uptake of distributed energy resources (DERs), including rooftop solar backed by battery storage systems, and electric vehicles (EVs). As electricity end-users become equipped with DERs, providing the means to both consume and generate electricity (i.e., an energy prosumer), new opportunities arise to participate in electricity markets \cite{azim2023dynamic, salehi2023competitive}. To efficiently operate distribution networks with growing populations of prosumers, new market mechanisms that support both energy trading and voltage regulation at the distribution-level are required \cite{meng2024holistic}.

An energy trading market operator in a competitive market determines the price for electricity, and sets the price according to the law of supply and demand \cite{arrow1954existence, mas1995microeconomic, salehi2024competitive}. In a competitive energy market, prosumers are considered small relative to the market size, and accordingly, individual decisions of prosumers do not influence the price \cite{mas1995microeconomic, salehi2022social}. Once the market operator sets the electricity price, each prosumer seeks to maximize their payoff \cite{alfaverh2023dynamic}, a combination of the cost of electricity consumption and income from electricity trading. The market reaches a \textit{competitive equilibrium} when no participant has an incentive to change their decision, and total supply matches total demand \cite{li2020transactive, salehi2022infinite, salehi2021social}. 

To efficiently operate an electric grid, physics-based constraints, including voltage and thermal limits, must also be considered in the design of a new market mechanism. Distribution Network Service Providers (DNSPs) typically impose static limits on customer imports and exports to ensure feeder-level voltage and thermal limits remain within permissible operating envelopes. However, static import and export limits are conservative, as they do not reflect time-varying grid conditions \cite{lankeshwara2023time}. Accordingly, DNSPs are increasingly adopting flexible grid connection agreements with prosumers, supporting operations with Dynamic Operating Envelopes (DOEs) \cite{Report2}. Specifically, a DOE is computed by a DNSP and is communicated to a prosumer --- providing a range of permissible active and reactive power injections (or demand) at a given time \cite{mahmoodi2023capacity}. 

Ancillary services are another element of an electricity market, improving the reliability and security of the electricity supply.  Ancillary services have typically been provided by large, centralized generators. As we transition to renewable-based grid operations, ancillary services are increasingly being delivered by a wider range of participants, including prosumers with DERs \cite{li2025energy}. Key ancillary services
include frequency regulation \cite{li2024incentive}, voltage support \cite{meng2024holistic}, and black-start capability \cite{aguilar2022exploring}. That is, within DOE limits, prosumers can contribute to both energy trading and ancillary services. However, limited DOE implementation and insufficient incentives for prosumer-based reactive power regulation potentially reduce the efficiency of existing market operations \cite{jiang2025bargaining, zou2024inverter}. 

{In our prior work \cite{salehi2026social}, we designed peer-to-peer (P2P) energy markets with active power trading in the absence of DOEs. We showed that the associated market equilibrium is a competitive equilibrium that maximizes social welfare, defined as the sum of the utilities of all customers.} In this paper, {we extend \cite{salehi2026social} to enable both active and reactive power trading among prosumers while achieving a balance between electricity supply and demand.  In addition,} we ensure that grid constraints are satisfied by localizing feeder-level grid constraints into prosumer-level DOE constraints.  To improve the overall utilization of the feeder, we design a market-mechanism whereby prosumers trade any excess DOE capacity, which would otherwise remain unused. 
Finally, we show that the proposed markets can be represented as a standard game, where prosumers are the players and a fictitious price-setting player determines the price, with the solution concept being a Nash equilibrium.
{Our preliminary work \cite{Salehi_Uniform} on real power trading with DOEs is also extended in this paper to include both real and reactive power trading}.

The rest of the paper is organized as follows. Section \ref{sec:preliminaries} introduces the prosumer dynamics and DOE constraints. Section \ref{sec:reactive} formulates the electricity market with reactive power trading to operate within feeder-level constraints. Section \ref{sec:Nash} investigates the relation of the proposed market with the Nash equilibrium of a standard game. Section \ref{sec:Numerical Results} presents numerical simulation results and Section \ref{sec:conclusion} concludes the paper.

\section{Preliminaries} \label{sec:preliminaries}

We consider an electrical distribution feeder with a radial topology connecting 
$N$ prosumers, each equipped with rooftop solar backed by battery storage. We consider the case where the electrical feeder operates in an island mode to support the blackstart capability of the bulk grid, i.e., all demand is met through local generation and P2P trading. We assume each prosumer has a home energy management system (EMS) that manages the rate of electricity consumption and energy trading according to the preferences of the prosumer. Within the EMS, the prosumer preferences are formulated into mathematical expressions termed the payoff, which is defined as the sum of the consumption utility (satisfaction level) and the income or expenditure from energy trading.

Each prosumer $i$, indexed in the set $\mathcal{N} = \{1, \dots, N\}$, has uncontrollable loads such as lighting, and controllable loads such as EVs and home battery storage systems. The net supply of each prosumer $i$ is defined as the difference between the rooftop solar generation and uncontrollable load and is denoted by $a_i(t) \in \mathbb{R}$ (kW). Controllable loads are modeled by linear difference equations 
\begin{equation}\label{eq_state}
	\mathbf x_i(t+1) = \mathbf A_i \mathbf x_i(t) + \mathbf B_i \mathbf u_i(t), \quad t \in \mathcal{T},
\end{equation}
where $\mathbf x_i(t) \in \mathbb{R}^n$ is the dynamical state (e.g., the state of charge (SoC) of a battery), $\mathbf x_i(0) \in \mathbb R^n$ is the initial state (e.g., the initial SoC of a battery), and $\mathbf u_i(t) \in \mathbb R^m$ is the control input (e.g., the charge/discharge rate of a battery). Also, $\mathbf A_i \in \mathbb R^{n \times n}$ and $\mathbf B_i \in \mathbb R^{n \times m}$ are fixed matrices. States and control inputs are constrained by 
$	\underline{\mathbf x}_i \leq \mathbf x_i(t) \leq \overline{\mathbf x}_i$ and $  \underline{\mathbf u}_i\leq \mathbf u_i(t) \leq \overline{\mathbf u}_i$, 
where $\underline{\mathbf x}_i \in \mathbb{R}^{n}$ and $\underline{\mathbf u}_i \in \mathbb{R}^{m}$ denote lower bounds, and $\overline{\mathbf x}_i \in \mathbb{R}^{n}$ and $\overline{\mathbf u}_i \in \mathbb{R}^{m}$  denote upper bounds.  

Prosumer preferences for managing controllable loads are expressed by  utility functions $f_i(\mathbf x_i(t), \mathbf u_i(t)): \mathbb{R}^{n} \times \mathbb{R}^{m} \mapsto \mathbb{R}$ that represent their satisfaction as a result of taking the control action $\mathbf u_i(t)$ and reaching the state $\mathbf x_i(t)$. For example, if the SoC of the battery is close to $80 \%$ and the number of charge/discharge cycles are relatively low, then the prosumer potentailly achieves a high level of satisfaction and a high $f_i(\cdot)$. Similarly, the terminal utility of a prosumer at the terminal time step $T$ is denoted by $\phi_i(\mathbf x_i(T)): \mathbb{R}^{n} \mapsto \mathbb{R}$, which depends on their satisfaction as a result of reaching the terminal state $\mathbf x_i(T)$. 
The average active power consumed or withdrawn as a result of the control action $\mathbf u_i(t)$ is denoted by $h_i(\mathbf u_i(t)): \mathbb{R}^{m} \mapsto \mathbb{R}$ (kW). At each time step $t$, if an agent $i$ has a supply shortage, i.e., $a_i(t) - h_i(\mathbf u_i(t))<0$, it can purchase the deficit from others; conversely, if it has a surplus, it can sell it to others through a trading decision variable $p_i(t) \in \mathbb{R}$ (kW) such that $p_i(t) \leq a_i(t) - h_i(\mathbf u_i(t))$\footnote{This inequality applies to both surplus and deficit cases.}. Active power injection is also constrained by inverter limits such that $\underline{ p}_i \leq  p_i(t) \leq \overline{ p}_i$, where  $\underline{ p}_i$ and  $\overline{ p}_i$ denote the lower and upper bound inverter limits, respectively.

Behind-the-meter inverters with reactive power regulation capability can support voltage regulation grid services. The injected or absorbed reactive power by prosumer $i$ is denoted by $q_i(t) \in \mathbb R$ (kVar) such that $	\underline{q}_i \leq q_i(t) \leq \overline{ q}_i$, where $\underline{q}_i$ and $\overline{ q}_i$ represent the lower and upper bound inverter limits, respectively. All inverter-based voltage regulation across a distribution feeder must comply with physics-based grid limits.
\subsection{Electric Grid  Constraints}\label{section:microgrid_model}
Denote $\mathbf p(t)=(p_1(t), \dots, p_N(t))^\top$ and $\mathbf q(t)=(q_1(t), \dots, q_N(t))^\top$ as the vectors of active and reactive power injections of all prosumers at time step $t \in \mathcal{T}$, respectively. We represent grid constraints in the separable form {as done in} \cite{mahmoodi2023capacity}
\begin{equation}\label{eq:constraints}
    \mathbf F_t(\mathbf p(t), \mathbf q(t))=\sum_{i=1}^N \mathbf g_{it}( p_i(t),  q_i(t)) \leq \boldsymbol \nu(t), \quad t \in \mathcal{T},
\end{equation}
where $\mathbf F_t: \mathbb{R}^{N} \times \mathbb{R}^{N} \mapsto \mathbb{R}^M$ is an affine vector function, and $M$ denotes the total number of grid constraints at time step $t \in \mathcal{T}$. Each function $\mathbf g_{it}: \mathbb{R}\times \mathbb R \mapsto \mathbb{R}^M$ is a local affine vector function associated with prosumer $i$ capturing its contribution to the $M$ global grid constraints at time step $t \in \mathcal{T}$. These functions are determined by grid characteristics such as line impedances and topology.
The vector $\boldsymbol \nu(t) \in \mathbb{R}^M$ specifies the bounds on the grid constraints, derived by the physical and operational limits of the network. Two examples of constraints expressed in the form of \eqref{eq:constraints} include voltage and thermal limits \cite{mahmoodi2023capacity}. Details of voltage constraints are provided in \cite{mahmoodi2023capacity, Salehi_Uniform}.

Enforcing the grid constraints in \eqref{eq:constraints} requires assigning feasible operating limits to individual prosumers. 
Static export limits or fixed power factor settings are often too conservative, as they fail to reflect the time-varying network conditions. DOEs have emerged as a promising approach to address this challenge.  
\begin{definition}[as in \cite{mahmoodi2023capacity}]
    A DOE is a nonempty, closed, and convex set $\Omega_{it}\subseteq \mathbb{R}^2$ in the $(p_{i}(t), q_i(t))$ space, that specifies an admissible region of active and reactive power injections for prosumer $i$ at time $t$, that the grid can accommodate without violating operational and physical constraints.
\end{definition}

Using the right-hand side decomposition (RHSD) method \cite{konnov2014right}, the right-hand side of \eqref{eq:constraints} can be expressed as
$    \boldsymbol \nu(t) = \sum_{i=1}^N \mathbf w_i(t)
$, where $\mathbf w_i(t) \in \mathbb{R}^M$ represents the share of the $M$ grid constraints assigned to prosumer $i$ at time step $t$. Consequently, the grid constraints in \eqref{eq:constraints} can be written as 
$
    \sum_{i=1}^N \mathbf g_{it}( p_i(t),  q_i(t)) \leq \sum_{i=1}^N \mathbf w_i(t)
$,
which, when decomposed, yields the DOEs
$
    \Omega_{it}=\big\{(p_{i}(t), q_i(t)) \in \mathbb{R}^2: \mathbf g_{it}( p_i(t),  q_i(t)) \leq \mathbf w_i(t)\big\}$
for each prosumer $i\in \mathcal{N}$ at time step $t\in \mathcal{T}$ \cite{mahmoodi2023capacity}. There are infinitely many possible choices for the RHSD and DOEs. In this paper, we focus on an allocation that maximizes the potential active power injection to the grid, by solving
\begin{equation}\label{eq:DOE_2}
		\begin{aligned}
			\max_{\mathbf w(t), \mathbf p(t), \mathbf q(t)} \quad &  \sum_{i=1}^{N} \max\{0, p_i(t)\} - \epsilon O(\mathbf w(t)) \\
			{\rm s.t.} 
			\quad &
			 \mathbf g_{it}( p_i(t),  q_i(t)) \leq \mathbf w_i(t) \\
    \quad & \sum_{i=1}^N \mathbf w_i(t) = \boldsymbol \nu(t),
    \,\,\, i \in \mathcal{N},
		\end{aligned}
	\end{equation}
where $\mathbf w(t)=(\mathbf w_1^\top(t), \dots, \mathbf w_N^\top(t))^\top$, and $\epsilon O(\mathbf w(t))$ is a regularization term with $\epsilon$ being a small positive number. The function $O(\mathbf w(t)): \mathbb{R}^{NM} \mapsto \mathbb{R}$ captures the preferences of the network service provider in allocating DOE shares among prosumers. To promote fairness, we consider $O(\mathbf w(t))=\|\mathbf w(t)- \overline{\mathbf{E}} \|$, where $\overline{\mathbf{E}} \in \mathbb R^{NM}, \overline{\mathbf{E}}=(\boldsymbol \nu^\top(t)/N, \dots, \boldsymbol \nu^\top(t)/N)^\top$ is the equality index. This choice of $O(\mathbf w(t))$ reflects that the network service provider tends to allocate close to equal DOEs to all prosumers \cite{mahmoodi2023capacity}. Solving \eqref{eq:DOE_2}, the network service provider determines the optimal DOEs and communicates them to each prosumer $i$ in the form of $\mathbf g_{it}( p_i(t),  q_i(t)) \leq \mathbf w_i^\ast(t)$.

\section{Market With Reactive Power Compensation}\label{sec:reactive}
Prosumers connected to the islanded feeder need to share their resources to meet their demand. In this section, we formulate a P2P energy market with reactive power trading that facilitates this interaction. Let $\lambda(t) \in \mathbb R$ denote the price for active power trading and $\gamma(t) \in \mathbb R$ the price for reactive power trading at time step $t$. In \cite{Salehi_Uniform}, we showed that prosumers can also trade their excess DOE limits, i.e., $\mathbf w_i^\ast(t)-\mathbf g_{it}( p_i(t),  q_i(t))$, through a decision variable $\boldsymbol l_i(t) \in \mathbb R^M$ at a price $\boldsymbol \beta(t) \in \mathbb R^M$, subject to $\boldsymbol l_i(t) \leq \mathbf w_i^\ast(t)-\mathbf g_{it}( p_i(t),  q_i(t))$, which mitigates the conservatism introduced by decomposing global grid constraints into local DOEs. In the following, we design a market that allows  real and reactive power trading as well as DOE limit trading. 

Obtaining DOE limits  $\mathbf{w}_i^\ast(t)$ from \eqref{eq:DOE_2},
 the coordinator determines equilibrium prices for active power trading $\lambda^\ast(t)$, reactive power trading $\gamma^\ast(t)$, and DOE limit trading $\boldsymbol \beta^\ast(t)$ that clear the market for $t \in \mathcal{T}$, and broadcasts them to participants. Upon receiving the equilibrium prices, prosumers engage in a P2P competitive market to maximize their payoff, defined as the sum of their utility $f_i(\mathbf x_i(t), \mathbf u_i(t))$, terminal utility $\phi_i(\mathbf x_i(T))$, and the income $\lambda (t)p_i(t)+ \gamma(t)q_i(t)+\boldsymbol \beta(t)\cdot \boldsymbol l_i(t)$ from transactions. 
 \begin{figure}[!t]
	\centering
	\includegraphics[width=3.35 in]{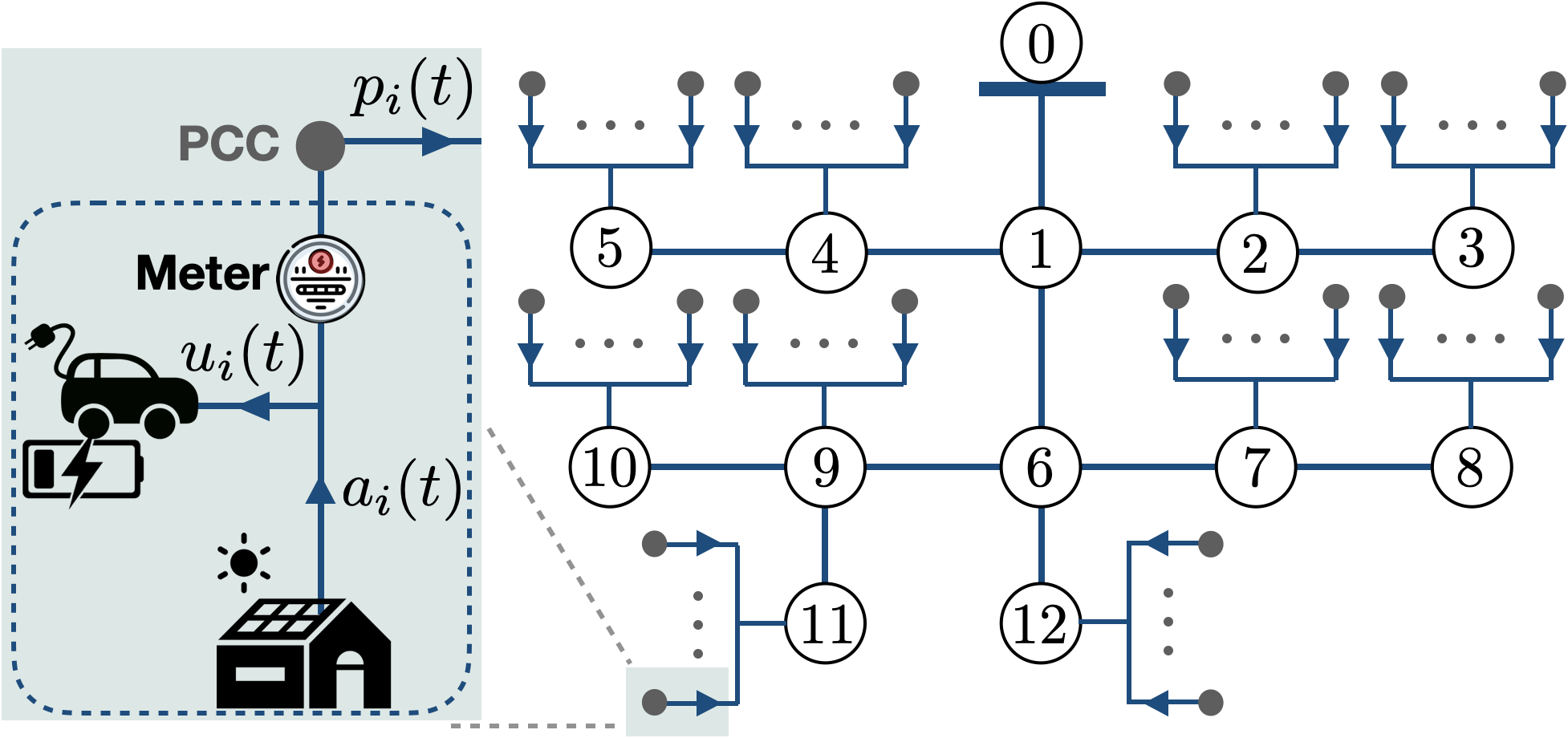}
	\caption{Modified IEEE 13-node test feeder. Left: A prosumer $i$ is zoomed in.}
	\label{fig_feeder}
\end{figure}
 
 The market reaches an equilibrium when no participant has an incentive to change their decision and total supply equals total demand, i.e., the market clears. This equilibrium is called a \textit{competitive equilibrium}, defined as follows. Let $\mathbf p_i=( p_i(0), \dots,  p_i(T-1))^\top$, $\mathbf q_i=( q_i(0), \dots,  q_i(T-1))^\top$, $\mathbf L_i=(\boldsymbol l_i^\top(0), \dots, \boldsymbol l_i^\top(T-1))^\top$, and $\mathbf U_i=(\mathbf u_i^\top(0), \dots, \mathbf u_i^\top(T-1))^\top$ denote, respectively, the vectors of active power injections,  reactive power injections, traded DOE limits, and control inputs of prosumer $i$ over the entire time horizon.  
 \begin{definition}\label{def5}
      Optimal decisions $(\mathbf u_1^\ast(t),  p_1^\ast(t), q_1^\ast(t), \\ \boldsymbol l_1^\ast(t), \dots, \mathbf u_N^\ast(t),  p_N^\ast(t), q_N^\ast(t), \boldsymbol l_N^\ast(t))$, along with optimal prices $ \lambda^\ast(t)$, $ \gamma^\ast(t)$, and $\boldsymbol \beta^\ast(t)$ for all $t \in \mathcal{T}$, form a \textit{competitive equilibrium} if the following statements hold.
     \begin{itemize}
         \item [(i)] Given $\lambda^\ast(t)$, $\gamma^\ast(t)$, $\boldsymbol \beta^\ast(t)$, and $\mathbf w_i^\ast(t)$, each prosumer $i \in \mathcal{N}$ maximizes their payoff at $(\mathbf u_i^\ast(t), p_i^\ast(t), q_i^\ast(t),  \boldsymbol l_i^\ast(t))$ for $t \in \mathcal{T}$ as a solution to       
         \begin{equation}\label{eq_competitive_7}
			\begin{aligned}
				\max_{{\mathbf U_i}, \mathbf p_i, \mathbf q_i, \mathbf L_i} \quad &  \sum_{t=0}^{T-1} \Big(f_i(\mathbf x_i(t), \mathbf u_i(t))  +   \lambda^\ast(t) p_i(t) \\ &{+\gamma^{\ast}(t)  q_i(t) +\boldsymbol  \beta^\ast(t) \cdot \boldsymbol l_i(t) }\Big) + {\phi_i(\mathbf x_i(T)}) \\
				{\rm s.t.} 
				\quad &  \mathbf x_i(t+1)= \mathbf A_i \mathbf x_i(t)+ \mathbf B_i \mathbf u_i(t),  \\
				\quad &  p_i(t) \leq a_i (t)- h_i(\mathbf u_i(t)),  \\
				\quad & {\boldsymbol l_i(t) \leq \mathbf w_i^\ast(t)-\mathbf g_{it}( p_i(t),  q_i(t))}, \\
				\quad &
    \underline{\mathbf x}_i \leq \mathbf x_i(t) \leq \overline{\mathbf x}_i, 
   \,\,\, \underline{\mathbf u}_i\leq \mathbf u_i(t) \leq \overline{\mathbf u}_i, \\
				\quad & \underline{ p}_i \leq  p_i(t) \leq \overline{ p}_i, 
   \,\,\, \underline{ q}_i\leq  q_i(t) \leq \overline{ q}_i,
				\,\,\,  t \in \mathcal{T}.
			\end{aligned}
		\end{equation}

  \item [(ii)] All traded power and all traded DOE limits are balanced at each time step; that is,
  \begin{equation}\label{equality_const}
      \sum_{i=1}^N p_i^\ast(t)=0, \,\,\, \sum_{i=1}^N q_i^\ast(t)=0, \,\,\, \sum_{i=1}^N \boldsymbol l_i^\ast(t)=0, \,\,\, t\in \mathcal{T}.
  \end{equation}
     \end{itemize}
 \end{definition}

 From microeconomics theory, a market outcome is Pareto efficient if there is no waste in the allocation of resources. Mathematically, a market outcome $(\mathbf u_1^\star(t),  p_1^\star(t), q_1^\star(t), \boldsymbol l_1^\star,  \dots, \mathbf u_N^\star(t),  p_N^\star(t), q_N^\star(t), \boldsymbol l_N^\star)$ for all $t \in \mathcal{T}$ is efficient if it maximizes the social welfare of the  society,  defined as the sum of all participant utilities $f_i(\cdot)$ and $\phi_i(\cdot)$, subject to  local and global  constraints. Such an operating point  $(\mathbf u_1^\star(t),  p_1^\star(t), q_1^\star(t), \boldsymbol l_1^\star,  \dots, \mathbf u_N^\star(t),  p_N^\star(t), q_N^\star(t), \boldsymbol l_N^\star)$ for all $t \in \mathcal{T}$ is referred to as a \textit{social welfare maximization solution} and solves
 \begin{equation}\label{eq4}
		\begin{aligned}
			\max_{{\mathbf U}, \mathbf P, \mathbf Q, \mathbf L} \quad &  \sum_{i=1}^{N}\Big(\sum_{t=0}^{T-1} f_i(\mathbf x_i(t), \mathbf u_i(t))+ {\phi_i(\mathbf x_i(T))} \Big)\\
			{\rm s.t.} 
			\quad &  \mathbf x_i(t+1)= \mathbf A_i \mathbf x_i(t)+ \mathbf B_i \mathbf u_i(t),  \\
			\quad &  p_i(t) \leq a_i (t)- h_i(\mathbf u_i(t)),   \\
				\quad & {\boldsymbol l_i(t) \leq \mathbf w_i^\ast(t)-\mathbf g_{it}( p_i(t),  q_i(t))}, \\
			\quad &	\sum_{i=1}^N p_i(t) = 0, \,\,\, 	 \sum_{i=1}^N q_i(t) = 0, 
			\,\,\, 	\sum_{i=1}^N \boldsymbol l_i(t) = 0, 
    \\
				\quad &
    \underline{\mathbf x}_i \leq \mathbf x_i(t) \leq \overline{\mathbf x}_i, 
   \,\,\, \underline{\mathbf u}_i\leq \mathbf u_i(t) \leq \overline{\mathbf u}_i, \\
				\quad & \underline{ p}_i \leq  p_i(t) \leq \overline{ p}_i, 
   \,\,\, \underline{ q}_i\leq  q_i(t) \leq \overline{ q}_i,
				\,\,\,   i \in \mathcal{N}, \,\,\, t \in \mathcal{T},
		\end{aligned}
	\end{equation}
where $\mathbf P=(\mathbf p_1^\top, \dots, \mathbf p_N^\top)^\top$, 
  $\mathbf Q=(\mathbf q_1^\top, \dots, \mathbf q_N^\top)^\top$, $\mathbf L=(\mathbf L_1^\top, \dots, \mathbf L_N^\top)^\top$, and $\mathbf U=(\mathbf U_1^\top, \dots, \mathbf U_N^\top)^\top$ denote, respectively,  the vectors of all active power injection, reactive power injection, traded DOE limits, and control inputs of all prosumers over the entire time horizon.

In the following theorem, we show that the market outcome in Definition \ref{def5} is Pareto efficient, maximizing \eqref{eq4}.

\begin{theorem} \label{theorem7}
    Suppose $f_i(\cdot)$, $\phi_i(\cdot)$, and $-h_i(\cdot)$ are concave functions for $i \in \mathcal{N}$, and Slater's condition holds for \eqref{eq_competitive_7} and \eqref{eq4}. Given feasible initial conditions $\mathbf x_i(0)$ for $i \in \mathcal{N}$, the following statements hold. 
    \begin{itemize}
        \item A competitive equilibrium in Definition \ref{def5} is equivalent to a social welfare maximization solution to \eqref{eq4}. 
        
        \item Let $-\alpha^\ast(t)$, $-\vartheta^\ast(t)$, and $-\boldsymbol \delta^\ast(t)$ be optimal dual variables corresponding to the balancing equality constraints $\sum_{i=1}^N p_i(t)=0$, $\sum_{i=1}^N q_i(t)=0$, and $\sum_{i=1}^N \boldsymbol l_i(t)=0$ in \eqref{eq4}, respectively, for $t\in \mathcal{T}$. Then,  equilibrium prices can be obtained as $\lambda^\ast(t)=\alpha^\ast(t)$, $\gamma^\ast(t)=\vartheta^\ast(t)$, and $\boldsymbol \beta^\ast(t)=\boldsymbol \delta^\ast(t)$ for $t\in \mathcal{T}$.
        \end{itemize}
\end{theorem}
\begin{proof}
See Appendix \ref{Appendix_Theorem7}.
\end{proof}

A procedure for implementing the proposed P2P market in Definition \ref{def5} is described in Algorithm \ref{algorithm1}.
\begin{algorithm}
\caption{Implementation of the P2P Market}
    \textbf{1:} DNSP determines DOEs by solving \eqref{eq:DOE_2} and sends $\mathbf w_i^\ast(t)$ to each prosumer $i$. 
    
    \textbf{2:} DNSP computes market prices $\lambda^\ast(t)$, $\gamma^\ast(t)$, and $\boldsymbol \beta^\ast(t)$ 
          as the optimal dual variables corresponding to the balancing constraints in \eqref{eq4}, and broadcasts them to prosumers.
          
    \textbf{3:} Prosumers form a P2P market and trade resources to maximize their payoff in \eqref{eq_competitive_7}.
\label{algorithm1}
\end{algorithm}

\begin{remark}
    Although Algorithm \ref{algorithm1} employs a centralized approach in Step 2 to determine market prices, decentralized alternatives such as the Alternating Direction Method of Multipliers (ADMM) or dual decomposition can also be applied, as in \cite{nguyen2020optimal, paudel2020decentralized}. In these methods, prosumers obtain resource prices in a decentralized manner, either with or without the supervision of a central coordinator. Following a procedure similar to \cite{paudel2020decentralized}, a fully decentralized and iterative scheme based on the dual decomposition method can be established for the P2P market proposed in Definition \ref{def5}. In a fully decentralized P2P framework, prosumers independently determine resource prices, and all information exchange occurs among neighboring participants without central coordination. {The exchanged information does not include any private information about the prosumers’ utility functions.} Such approaches preserve participant privacy.
\end{remark}

\section{Nash Equilibrium}\label{sec:Nash}
In this section,  we examine the connection between the proposed  market and a standard game. We show that the market outcome in Definition \ref{def5} can be implemented as a Nash equilibrium of a standard game with $N+1$ players. 

\begin{itemize}
    \item 
The first player is an aggregator who takes the role of a price player and regulates the price such that the market clears. Given $p_i(t)$, $q_i(t)$, and $\boldsymbol l_i(t)$ for $i \in \mathcal{N}$ and $t \in \mathcal{T}$, the price player obtains  uniform prices $\lambda(t)$, $\gamma(t)$, and $\boldsymbol \beta(t)$ as a solution to
\begin{equation}\label{game2}
\resizebox{\linewidth}{!}{$
\begin{aligned}
\max_{{\boldsymbol \lambda}, {\boldsymbol \gamma}, {\boldsymbol \beta}} \quad &  \sum_{i=1}^N \sum_{t=0}^{T-1} \Big(-\lambda(t)p_i(t) -\gamma(t)q_i(t)  
- \boldsymbol \beta(t) \cdot \boldsymbol{l}_i(t) \Big),
\end{aligned} $}
\end{equation}
where $\boldsymbol \lambda=(\lambda(0), \dots, \lambda(T-1))^\top$, $\boldsymbol \gamma=(\gamma(0), \dots, \gamma(T-1))^\top$, and $\boldsymbol \beta=(\boldsymbol \beta^\top(0), \dots, \boldsymbol \beta^\top(T-1))^\top$ are the vectors of all active power, reactive power, and DOE limit prices over the entire time horizon, respectively.

\item The last $N$ players are prosumers who maximize their payoff subject to the local constraints. Given prices $\lambda(t)$, $\gamma(t)$, and $\boldsymbol{\beta}(t)$ for $t \in \mathcal{T}$, each prosumer $i \in \mathcal{N}$ obtains $(\mathbf U_i, \mathbf p_i, \mathbf q_i, \mathbf L_i)$ as a solution to
\begin{equation}\label{game1}
\begin{aligned}
				\max_{{\mathbf U_i}, \mathbf p_i, \mathbf q_i, \mathbf L_i} & \quad   \sum_{t=0}^{T-1} f_i(\mathbf x_i(t), \mathbf u_i(t)) + {\phi_i(\mathbf x_i(T)}) \\ &+ \sum_{t=0}^{T-1} \Big( \lambda(t) p_i(t)+ \gamma(t) q_i(t)+\boldsymbol  \beta(t) \cdot \boldsymbol l_i(t)\Big) \\
				{\rm s.t.} 
				\quad &  \mathbf x_i(t+1)= \mathbf A_i \mathbf x_i(t)+ \mathbf B_i \mathbf u_i(t),  \\
				\quad &  p_i(t) \leq a_i (t)- h_i(\mathbf u_i(t)),  \\
				\quad & \boldsymbol l_i(t) \leq \mathbf w_i^\ast(t)-\mathbf g_{it}( p_i(t), q_i(t)), \\
				\quad &
    \underline{\mathbf x}_i \leq \mathbf x_i(t) \leq \overline{\mathbf x}_i, 
   \,\,\, \underline{\mathbf u}_i\leq \mathbf u_i(t) \leq \overline{\mathbf u}_i, \\
				\quad & \underline{ p}_i \leq  p_i(t) \leq \overline{ p}_i, 
   \,\,\, \underline{ q}_i\leq  q_i(t) \leq \overline{ q}_i,
				\,\,\,  t \in \mathcal{T}.
			\end{aligned}
\end{equation}
\end{itemize}
The standard game in \eqref{game2}--\eqref{game1} stems from the generalized game introduced in \cite{arrow1954existence}, where the objective function of the price player indicates the ``law of supply and demand".
\begin{theorem}\label{theorem6}
    Let $(\mathbf U^\ast, \mathbf P^\ast, \mathbf Q^\ast, \mathbf L^\ast)$, along with  $ {\lambda}^\ast(t)$, $ \gamma^\ast(t)$, $ \boldsymbol \beta^\ast(t)$ for $t \in \mathcal{T}$,  be a Nash equilibrium for the standard game in \eqref{game2}--\eqref{game1}. Then  $(\mathbf U^\ast, \mathbf P^\ast, \mathbf Q^\ast, \mathbf L^\ast)$, along with  $ {\lambda}^\ast(t)$, $ \gamma^\ast(t)$, $ \boldsymbol \beta^\ast(t)$ for $t \in \mathcal{T}$, is equivalent to a competitive equilibrium as defined in Definition \ref{def5}.
\end{theorem}
\begin{proof}
    A Nash equilibrium is obtained when  $\sum_{i=1}^N p_i^\ast(t)= 0$, $\sum_{i=1}^N q_i^\ast(t)= 0$, and $\sum_{i=1}^N \boldsymbol l_i^\ast(t)=0$ for $t \in \mathcal{T}$, and therefore, the objective of the first player in \eqref{game2} has a finite solution. This is equivalent to the second condition in Definition \ref{def5}. Additionally, a Nash equilibrium is obtained when the objective of the last $N$ players in \eqref{game1} is maximized, which corresponds to the first condition in Definition \ref{def5}. Consequently, it is straightforward to show that $(\mathbf U^\ast, \mathbf P^\ast, \mathbf Q^\ast, \mathbf L^\ast)$, along with  $ {\lambda}^\ast(t), {\gamma}^\ast(t),  \boldsymbol \beta^\ast(t)$ for $t \in \mathcal{T}$, satisfies all the conditions in Definition \ref{def5}. Furthermore, all competitive equilibria in Definition \ref{def5} serve as a Nash equilibrium for the game defined by \eqref{game2}--\eqref{game1}.
\end{proof}

\section{Numerical Simulations}\label{sec:Numerical Results}
In this section, we benchmark the proposed market with reactive power trading against one without reactive power participation. For simplicity, we consider voltage constraints in the electric grid; extensions to include thermal constraints are straightforward. We use a modified single-phase representation of the IEEE 13-node test feeder, as shown in Fig. \ref{fig_feeder}, where the transformer between nodes 2 and 3, the switch between nodes 6 and 7, and all capacitor banks are omitted \cite{nimalsiri2021coordinated}. All simulations are conducted in Python 3.8 using CVXPY on a MacBook Pro with an Apple M1 chip and 16 GB of memory. {In particular, the nonconvex optimization problem in \eqref{eq:DOE_2} is solved using CVXPY.}

\subsection{Simulation Setup}
Apart from the feeder node, node 1, and node 6,  each node is connected to $30$ distribution aggregators. Each aggregator manages $11$ residential prosumers, each equipped with a grid-connected EV and a home battery storage system. The net power supply of each prosumer---defined as rooftop solar generation minus uncontrollable demand---is derived from metering data provided by Ausgrid \cite{Ausgrid1}, an Australian electricity distribution network provider. This dataset, originally collected in 2012, is scaled by a factor of 2 to reflect increased solar generation and electrification levels.
 EV models are from a diverse set documented in \cite{Battery1}, with battery capacities ranging from $0$ to $75$ (kWh). The corresponding home battery storage systems are assumed to have similar capacity characteristics. 
The charge/discharge efficiencies of EVs and home battery storage systems are assumed to be $\eta=0.9$. The system is simulated over a  $24$-hour (h) time horizon starting at midnight, with a sampling time of $\Delta=0.5$ (h). 

The aggregated behavior of each aggregator $i$, incorporating both EVs and home battery storage systems of the associated $11$ prosumers, is modeled using the linear difference equation $\mathbf x_i(t+1)=\mathbf x_i(t)+\eta \mathbf u_i(t)\Delta$, where $\mathbf x_i(t) \in \mathbb R^2$ denotes the SoC vector (kWh) of the aggregated EVs and residential battery storage systems, and $\mathbf u_i(t) \in \mathbb R^2$ represents the corresponding charge/discharge rate vector (kW) at each time step $t$. In both vectors, the first element corresponds to the aggregated EVs, and the second to the aggregated home battery storage systems. 
The initial SoC of each battery is assumed to follow a uniform distribution over the interval $[0.2, 0.5]$ of its respective capacity. Home battery storage systems remain grid-connected at all times, whereas EV arrival and departure times are drawn from survey data collected by the Victorian Department of Transport in Australia \cite{Mendeley1}. Prosumers are assumed to use level-$2$ chargers with a maximum charge rate of $6.6$ (kW) \cite{nimalsiri2021coordinated}, and a corresponding  maximum discharge rate of $-6.6$ (kW). To preserve battery health, the SoC of both EVs and home battery storage systems is constrained to remain between $20\%$ and $85\%$ of their respective capacities. Accordingly, the aggregated SoC and charging/discharging rates for each aggregator $i$ are bounded by $0.2 \mathbf C_i\leq \mathbf x_i(t)\leq 0.85 \mathbf C_i$  and $-6.6 \mathbf 1 \leq \mathbf u_i(t) \leq 6.6 \mathbf 1$, respectively, where  $\mathbf C_i \in \mathbb R^2$ denotes the vector of aggregated battery capacities (kWh) and $\mathbf{1} \in \mathbb R^2$ is a vector of ones. The first and second elements of $\mathbf C_i$ correspond to the aggregated EV and home battery capacities, respectively, of the prosumers managed by aggregator $i$. Prosumers have inverters with reactive power import and export limits of $0.33$ kVar. 

In this framework, each of the 
$N=300$ aggregators functions as a decision-maker, approximating the aggregate behavior of its associated prosumers. The utility function for aggregator $i \in \mathcal{N}$ is defined as a quadratic form\footnote{The utility function can take any form, the only requirement is that it is concave.} $f_i(\mathbf x_i(t), \mathbf u_i(t))=- \vartheta_1\| \mathbf u_i(t)\|^2- \vartheta_{2,t} \|(\mathbf x_i(t)-0.85 \mathbf C_i) \cdot (1,0)^\top\|^2 - \vartheta_{3,t} \|(\mathbf x_i(t)-0.85 \mathbf C_i) \cdot (0,1)^\top\|^2$ for $t \in \mathcal{T}=\{0, \dots, 47\}$, where $\|\cdot\|$ denotes the Euclidean norm. The terminal utility is given by $\phi_i(\mathbf x_i(T))=-\vartheta_4\|\mathbf x_i(T)-0.85 \mathbf C_i\|^2$,  with  $T=48$ as the terminal time step. The  coefficients  $\vartheta_1$ ($10^{-8}$\textcent/(kW)$^2$), $\vartheta_{2,t}$ ($10^{-8}$\textcent/(kWh)$^2$), $\vartheta_{3,t}$ ($10^{-8}$\textcent/(kWh)$^2$), and $\vartheta_4$ ($10^{-8}$\textcent/(kWh)$^2$) serve as regularization weights. For each aggregator $i$, the power consumed or withdrawn by the aggregated controllable loads (EVs and home battery storage systems) is modeled as $h_i(\mathbf u_i(t))=\mathbf u_i(t) \cdot \mathbf 1$. The feeder nominal voltage is set to $V_0=1$ (p.u.). In compliance with the ANSI C84.1 standard, voltage magnitudes must remain within $\pm5 \% $ of the nominal voltage. 

\subsection{Numerical Simulations}
 Considering voltage constraints, the DNSP obtains DOEs according to \eqref{eq:DOE_2}, and sends the obtained $\mathbf w_i^\ast(t)$ to each aggregator $i$ for $i \in \{1, \dots, 300\}$. The DNSP sends the active and reactive power prices, $\lambda^\ast(t)$ and $\gamma^\ast(t)$, and DOE limit prices $\boldsymbol \beta^\ast(t)$, which could be obtained solving \eqref{eq4}, to each aggregators. Each aggregator $i$ then decides on its electricity consumption and trading  to maximize its payoff according to \eqref{eq_competitive_7}. The active and reactive power prices are depicted in Fig. \ref{fig_power_price}, and the DOE limit prices are shown in Fig. \ref{fig_DOE_price}. 

Zero and negative reactive power prices in Fig. \ref{fig_power_price} correspond to cases where reactive power is in excess in the grid and must be absorbed by inverters. Similarly, zero prices for different elements of DOE limit trading in Fig. \ref{fig_DOE_price} indicate that the associated voltage constraints remain within their boundaries, whereas nonzero prices correspond to cases where the respective voltages reach the boundary.
 \begin{figure}[!t]
	\centering
	\includegraphics[width=2.7 in]{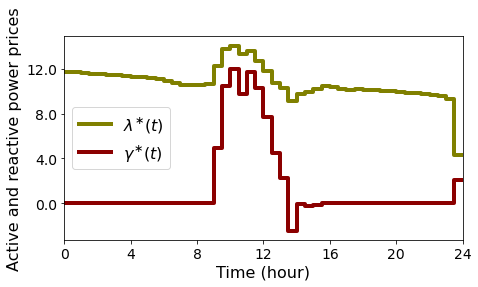}
	\caption{Active power prices $\lambda^\ast(t)$ (\textcent/kWh) and reactive power prices $\gamma^\ast(t)$ (\textcent/kVarh), over a 24-hour period (48 time steps of length $0.5$ hour).}
	\label{fig_power_price}
\end{figure}
\begin{figure}[!t]
	\centering
	\includegraphics[width=3 in]{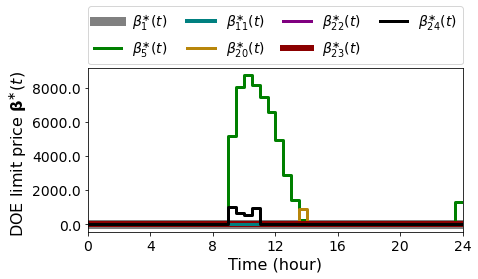}
	\caption{Some elements of the DOE limit price vector $\boldsymbol\beta^\ast(t)$ (in \textcent/(kV)$^2$) over a 24-hour period (48 time steps of length $0.5$ hour). The $k$-th element of  $\boldsymbol \beta^\ast(t)$ is denoted by $\beta^\ast_{k}(t)$. The other $17$ elements are close to zero and  not depicted.}
	\label{fig_DOE_price}
\end{figure}
 %
 %
 The nodal voltages are shown in Fig. \ref{fig_voltage} and compared with the case without reactive power injection (gray area). As observed, reactive power trading functions as an ancillary service market that improves voltage regulation, with voltages reaching their boundaries less frequently. Moreover, we observed that reactive power trading increases community profit, as the social welfare of the system rises by $1.6\%$ compared to the case without reactive power trading.
\begin{figure}[!t]
	\centering
	\includegraphics[width=3 in]{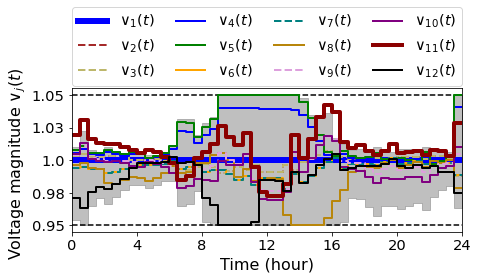}
	\caption{Nodal voltages $v_j(t)$ (p.u.) for each of the $12$ nodes, $j\in \{1, \dots, 12\}$. The gray area indicates the voltage range without reactive power trading across the grid.}
	\label{fig_voltage}
\end{figure}
\section{Conclusion}\label{sec:conclusion}
We considered a radial distribution feeder operating in island mode to support the blackstart capability of the bulk grid. We designed a peer-to-peer energy market to enable energy trading by prosumers with behind-the-meter rooftop solar backed by battery storage. Moreover, we design the market to incorporate voltage regulation across the feeder by way of inverter-based resources regulating reactive power flows to and from prosumers. Specifically, we decomposed power flow constraints across the feeder into localized customer-level limits, termed dynamic operating envelopes (DOEs). We showed that our proposed market maximizes social welfare and is equivalent to a standard game. We benchmarked our proposed market using the IEEE 13-node test feeder under binding voltage constraints. Our results demonstrate that reactive power trading functions as an ancillary service market, improving both the community’s welfare and the overall voltage profile of the feeder compared to the case without reactive power trading. Future work to consider heterogeneous prosumer utility functions is possible. 

\appendices
\section{Proof of Theorem \ref{theorem7}}\label{Appendix_Theorem7}
Let $( \mathbf U^\ast, \mathbf P^\ast, \mathbf Q^\ast, \mathbf L^\ast)$, along with ${\lambda}^\ast(t)$, ${\gamma}^\ast(t)$, and $ \boldsymbol \beta^\ast(t)$ for $t \in \mathcal{T}$,  be a competitive equilibrium. Merging optimization problems in \eqref{eq_competitive_7} for $i \in \mathcal{N}$ implies that $( \mathbf U^\ast, \mathbf P^\ast,  \mathbf Q^\ast, \mathbf L^\ast)$ maximizes the social welfare in \eqref{eq4}. 

   The proof of the reverse direction is based on strong duality, which follows from the satisfaction of Slater's condition. Considering the difference equation \eqref{eq_state}, the state $\mathbf x_i(t)$ evolves according to a linear combination of the initial state $\mathbf x_i(0)$ and the control input $\mathbf U_i$ such that
$
	\mathbf x_i(t)=\mathbf A_i^t \mathbf x_i(0)+ \sum_{j=0}^{t-1}{\mathbf A_i^{t-j-1} \mathbf B_i \mathbf u_i(j)}$ for  $t \in \{1, 2, ..., T\}
    $,
whose substitution into  $f_i(\cdot)$ and $\phi_i(\cdot)$ results in
$
		f_i(\mathbf x_i(t), \mathbf u_i(t)) = \tilde{f}_{i,t}(\mathbf U_i)$ and
$		\phi_i(\mathbf x_i(T)) = \tilde{\phi}_i(\mathbf U_i)
$.
The functions $\tilde{f}_{i,t}(\cdot)$ and $\tilde{\Phi}_i(\cdot)$ are concave as the composition of a concave function and an affine function. For $i \in \mathcal{N}$, denote $\mathbb{U}_i =\{ \mathbf U_i| \underline{\mathbf u}_i\leq \mathbf u_i(t) \leq \overline{\mathbf u}_i;  \underline{\mathbf x}_i \leq \mathbf x_i(t) \leq \overline{\mathbf x}_i; \underline{ p}_i \leq  p_i(t) \leq \overline{ p}_i; \underline{ q}_i\leq  q_i(t) \leq \overline{ q}_i, \text{for } t \in \mathcal{T}\}$,  which is a polyhedral set. For any $(\mathbf U, \mathbf P, \mathbf Q, \mathbf L)$ such that $\mathbf U_i \in \mathbb{U}_i$, the Lagrangian function of \eqref{eq4} is defined as
\begin{equation}\label{eq15}
	\begin{aligned}
	&\mathcal L(\mathbf U, \mathbf P, \mathbf Q, \mathbf L, \boldsymbol{\alpha}, \boldsymbol \vartheta, \boldsymbol{\delta}, \boldsymbol\Psi, \boldsymbol \Pi) = \\ &- \sum_{i=1}^{N}\Big(\sum_{t=0}^{T-1} \tilde{f}_{i,t}( \mathbf U_i)+\tilde{\phi}_i(\mathbf U_i) \Big) 
	+ \sum_{t=0}^{T-1}\alpha(t) \Big(-\sum_{i=1}^N p_i(t) \Big) \\&+ \sum_{t=0}^{T-1}\vartheta(t) \Big(-\sum_{i=1}^N q_i(t) \Big) 
	+ \sum_{t=0}^{T-1} \boldsymbol\delta(t) \cdot \Big(-\sum_{i=1}^N \boldsymbol l_i(t) \Big) 
	\\&+ \sum_{t=0}^{T-1}\sum_{i=1}^{N} \psi_{i}(t) \Big( p_i(t) +h_i(\mathbf u_i(t))-a_i(t)\Big),
 \\&+ \sum_{t=0}^{T-1}\sum_{i=1}^{N} \boldsymbol \pi_{i}(t) \cdot \Big( \boldsymbol l_i(t) +\mathbf g_i( p_i(t))-\mathbf w_i^\ast(t)\Big),
\end{aligned}
\end{equation}
where $\psi_{i}(t) \geq 0$, $\boldsymbol \pi_i(t) \geq 0$, $\boldsymbol \alpha =(\alpha(0), \dots, \alpha({T-1}))^\top$, $\boldsymbol \vartheta =(\vartheta(0), \dots, \vartheta({T-1}))^\top$, $\boldsymbol \delta =(\boldsymbol\delta^\top(0), \dots, \boldsymbol\delta^\top({T-1}))^\top$,  $\boldsymbol\Psi_i =(\psi_{i}(0), \dots, \psi_{i}(T-1))^\top$, and $ {\boldsymbol\Psi} = ( {\boldsymbol\Psi}_1^\top, \dots,  {\boldsymbol\Psi}_N^\top)^\top$. We define $\boldsymbol\Pi_i$ and  $\boldsymbol\Pi$  in a similar way. 
The Lagrangian  \eqref{eq15}  is separable such that
\begin{multline}\label{eq16}
	\mathcal L(\mathbf U, \mathbf P, \mathbf Q, \mathbf L, \boldsymbol{\alpha}, \boldsymbol \vartheta, \boldsymbol \delta, \boldsymbol\Psi, \boldsymbol \Pi) \\ = \sum_{i=1}^{N} \mathcal L_i(\mathbf U_i, \mathbf p_i, \mathbf q_i, \mathbf L_i, \boldsymbol{\alpha}, \boldsymbol{\vartheta}, \boldsymbol{\delta}, \boldsymbol\Psi_i, \boldsymbol \Pi_i),
\end{multline}
where
\begin{equation*}\label{eq9}
	\begin{aligned}
	\mathcal L_i &(\mathbf U_i, \mathbf p_i, \mathbf q_i, \mathbf L_i, \boldsymbol{\alpha}, \boldsymbol{\vartheta}, \boldsymbol{\delta}, \boldsymbol\Psi_i, \boldsymbol \Pi_i)
	=- \sum_{t=0}^{T-1} \Big(\tilde{f}_{i,t}(\mathbf U_i) \\& +\alpha(t)  p_i(t) + \vartheta(t)  q_i(t)  
	  + \boldsymbol \delta(t) \cdot \boldsymbol l_i(t)\Big) - \tilde{\phi}_i(\mathbf U_i) 
	\\&+  \sum_{t=0}^{T-1} \psi_{i}(t) \Big( p_i(t)+h_i(\mathbf u_i(t))-a_i(t) \Big) \\
 &+ \sum_{t=0}^{T-1}\boldsymbol \pi_{i}(t) \cdot \Big( \boldsymbol l_i(t) +\mathbf g_i( p_i(t))-\mathbf w_i^\ast(t)\Big).
\end{aligned}
\end{equation*}
Let $(\mathbf U^\star, \mathbf P^\star, \mathbf Q^\star, \mathbf L^\star)$ be an optimal primal solution to  \eqref{eq4}, and  $(-\boldsymbol{\alpha}^\ast, -\boldsymbol{\vartheta}^\ast, -\boldsymbol \delta^\ast, \boldsymbol{\Psi}^\ast , \boldsymbol{\Pi}^\ast)$ be an optimal dual solution. Given that Slater's condition is satisfied, strong duality implies 
\begin{multline}\label{eq_primal_4}
	(\mathbf U^\star, \mathbf P^\star, \mathbf Q^\star, \mathbf L^\star) \\ \in \arg \min_{\mathbf U, \mathbf P, \mathbf Q, \mathbf L} \mathcal L(\mathbf U, \mathbf P, \mathbf Q, \mathbf L, \boldsymbol{\alpha}^\ast, \boldsymbol{\vartheta}^\ast, \boldsymbol \delta^\ast, \boldsymbol\Psi^\ast, \boldsymbol\Pi^\ast),
\end{multline}
for $ \mathbf U_i \in \mathbb{U}_i$. 
Considering \eqref{eq16} and \eqref{eq_primal_4}, there holds
\begin{multline}\label{eq13}
	(\mathbf U_i^\star,  \mathbf p_i^\star, \mathbf q_i^\star, \mathbf L_i^\star) \\ \in \arg \min_{\mathbf U_i, \mathbf p_i, \mathbf q_i, \mathbf L_i} \mathcal L_i( \mathbf U_i,  \mathbf p_i, \mathbf q_i, \mathbf L_i, \boldsymbol \alpha^\ast, \boldsymbol \vartheta^\ast, \boldsymbol \delta^\ast, \boldsymbol{\Psi}_i^\ast, \boldsymbol{\Pi}_i^\ast),
\end{multline}
for $ \mathbf U_i \in \mathbb{U}_i$. 
In addition, strong duality implies
\begin{equation*}
\resizebox{\linewidth}{!}{$
 \begin{aligned}
	&(\boldsymbol{\alpha}^\ast, \boldsymbol{\vartheta}^\ast, \boldsymbol \delta^\ast, \boldsymbol{\Psi}^\ast, \boldsymbol{\Pi}^\ast) \\ &\in \arg \max_{\boldsymbol \alpha, \boldsymbol \vartheta, \boldsymbol \delta, \boldsymbol{\Psi}, \boldsymbol{\Pi}} \min_{\mathbf U, \mathbf P, \mathbf Q, \mathbf L} 	\mathcal L(\mathbf U, \mathbf P, \mathbf Q, \mathbf L, \boldsymbol{\alpha}, \boldsymbol{\vartheta}, \boldsymbol \delta, \boldsymbol\Psi, \boldsymbol\Pi) 
	\\&= \arg \max_{\boldsymbol \alpha, \boldsymbol \vartheta, \boldsymbol \delta, \boldsymbol{\Psi}, \boldsymbol{\Pi}} \min_{\mathbf U, \mathbf P, \mathbf Q, \mathbf L} 	\sum_{i=1}^{N} \mathcal L_i(\mathbf U_i, \mathbf p_i, \mathbf q_i, \mathbf L_i, \boldsymbol{\alpha}, \boldsymbol{\vartheta}, \boldsymbol{\delta}, \boldsymbol\Psi_i, \boldsymbol\Pi_i),
 \end{aligned} $}
\end{equation*} 
and, therefore, 
$
	(\boldsymbol{\Psi}_i^\ast, \boldsymbol{\Pi}_i^\ast)  \in  \arg \max_{\boldsymbol \Psi_i, \boldsymbol \Pi_i} \min_{\mathbf U_i, \mathbf p_i, \mathbf q_i, \mathbf L_i} \\ \mathcal L_i(\mathbf U_i, \mathbf p_i, \mathbf q_i, \mathbf L_i, \boldsymbol{\alpha}^\ast, \boldsymbol{\vartheta}^\ast, \boldsymbol{\delta}^\ast, \boldsymbol\Psi_i, \boldsymbol\Pi_i),
    $
where the primal and dual variables are in their respective domains. Considering  $\lambda^\ast(t)=\alpha^\ast(t)$, $\gamma^\ast(t)=\vartheta^\ast(t)$, and $\boldsymbol \beta^\ast(t)=\boldsymbol \delta^\ast(t)$ for $t \in \mathcal{T}$, the function $ \mathcal L_i(\mathbf U_i, \mathbf p_i,  \mathbf q_i, \mathbf L_i, \boldsymbol{\alpha}^\ast, \boldsymbol{\vartheta}^\ast, \boldsymbol{\delta}^\ast, \boldsymbol\Psi_i, \boldsymbol\Pi_i)$ is the Lagrangian  of \eqref{eq_competitive_7}. 
Therefore, $(\mathbf \Psi_i^\ast, \mathbf \Pi_i^\ast)$ is an optimal dual solution of \eqref{eq_competitive_7}.  Given that
Slater's condition is satisfied, strong duality implies that according to  \eqref{eq13}, $(\mathbf U^\star, \mathbf P^\star, \mathbf Q^\star, \mathbf L^\star)$ is an optimal primal solution of \eqref{eq_competitive_7}  which satisfies \eqref{equality_const}. 
Consequently, $( \mathbf U^\star, \mathbf P^\star, \mathbf Q^\star, \mathbf L^\star)$, along with $\lambda^\ast(t)=\alpha^\ast(t)$, $\gamma^\ast(t)=\vartheta^\ast(t)$, and $\boldsymbol \beta^\ast(t)=\boldsymbol \delta^\ast(t)$ for $t \in \mathcal{T}$, forms a competitive equilibrium.


\end{document}